\documentclass[12pt]{iopart}

\usepackage{graphicx}
\usepackage{float}
\usepackage{amssymb}
\usepackage{amsfonts}
\usepackage{amsthm}
\usepackage[all,cmtip]{xy}
\usepackage{mathrsfs}

\begin{document}

\title[Fractional Moment Methods for Anderson Localization with SAW Representation]{Fractional Moment Methods for Anderson Localization with SAW Representation}

\author{Fumika Suzuki}

\address{Department of Physics \& Astronomy, The University of British Columbia, Vancouver, BC V6T 1Z1, Canada}
\ead{fumika@physics.ubc.ca}
\begin{abstract}
The Green's function contains much information about physical systems. Mathematically, the fractional moment method (FMM)  developed by Aizenman and Molchanov connects the Green's function and the transport of electrons in the Anderson model. Recently, it has been discovered that the Green's function on a graph can be represented using self-avoiding walks on a graph, which allows us to connect localization properties in the system and graph properties. We discuss FMM in terms of the self-avoiding walks on a general graph, the only general condition being that the graph has a uniform bound on the vertex degree. 

\end{abstract}

\maketitle
\newtheorem{theorem}{Theorem}[section]
\newtheorem{lemma}[theorem]{Lemma}
\newtheorem{proposition}[theorem]{Proposition}
\newtheorem{corollary}[theorem]{Corollary}
\newtheorem{definition}[theorem]{Definition}
\newtheorem{example}[theorem]{Example}

\section{Introduction}

Rigorous studies of Anderson localization in the mathematical context
started in the 1970s. So far, there exist several methods invented to prove
Anderson localization and two methods provide proofs of Anderson localization
in arbitrary dimension, not only one dimension. They are multiscale
analysis (MSA) by Fr\"{o}hlich and Spencer \cite{Fr} (for the survey of MSA, refer to \cite{Kl}) and the fractional moment
method (FMM) by Aizenman and Molchanov \cite{Aizen1,Aizen}. Although
MSA can handle more situations of the Anderson model than FMM can,
FMM is a simpler method and gives us stronger results on dynamical localization. In this paper, we deal with Anderson localization using FMM. 

This paper is very closely related to the papers \cite{Hun,Tau} which have used the self-avoiding walk (SAW) representation of Green's function in proving localization properties. \cite{Hun} studied the Anderson model on $\mathbb{Z}^{n}$. Here, we consider more general graphs, the only general condition being that the graph has a uniform bound on the vertex degree. \cite{Tau} does not assume that the graph has uniformly bounded vertex degree and so a larger class of graphs have been studied. As a result, \cite{Tau} put much of its focus on issues related to unboundedness of the Hamiltonian. This only allows for a weaker form of dynamical localization than our definition (1) in the section 2. 

Although Aizenman and Molchanov noted already in the pioneer papers \cite{Aizen1,Aizen} that the FMM applies also to uniformly bounded graphs and gives localization in the large disorder regime, the following are the important objectives of this paper.\\

(i) It seems that many studies of Anderson localization in the mathematics community and in the physics community have been developed independently without communicating each other. While the connection between self-avoiding walks and Anderson localization has been observed before and is known to specialists in the field of mathematical physics, there is still need for making it more broadly known to researchers in other fields, in particular, theoretical physicists. For those readers, this paper provides a proof of localization which is transparent and intuitive and uses a minimal amount of mathematical technicalities.

(ii) The main result, Theorem 5.1 provides a better understanding of how localization properties depend not only on the amount of disorder in the random potential, but also on graph properties such as the number of self-avoiding walks of given length and the volume growth of the graph. It is written in the simple form so that it can be used for future investigations not only in mathematical physics, but also in other fields such as theoretical physics.

\section{The Anderson Model}
In this section, we introduce the Anderson model which describes the motion of an electron in a disordered system. There exist many good reviews of the Anderson model such as \cite{Hun, Sto} and we follow their notations in this paper. In physics, it is common to model the system by a lattice $\mathbb{Z}^{n}$ or a graph $\mathbb{G}$. In this paper, we deal with the Anderson model on a graph $\mathbb{G}$ with the following assumptions.\\

Let $\mathbb{G} = (V, E)$ be a graph with vertices $V$ and edges $E$. We assume $\mathbb{G}$ is a connected graph where the number of edges between any pair of vertices is either one or zero. We write $v\sim w$ if an edge connects the vertices $v$ and $w$. Let $N(v)$ be the degree of a vertex $v\in V$. i.e., $N(v) := \# \{w \in V : w\sim v\}$. We assume that the degree of a vertex is bounded above by some constant $N$, $N(v) \leq N < \infty$ for all $v \in V$. $d(v,w)$ is the graph distance from $v$ to $w$, which is the minimum number of edges from $v$ to $w$ on $\mathbb{G}$. \\

$\mathcal{W} (v,w)$ is the set of self-avoiding walks (sequences of vertices) $[v , v_1, \ldots , v_{d(v,w)}]$ with $d(v,w)$ steps starting at $v_0 =v$. The walks need not end at $w$.

 Furthermore, $\mathcal{W}' (v,w) $ is the set of self-avoiding walks $[v,v_1,\ldots,v_{d(v,w)}]$ with $d(v,w)$ steps starting at $v_0 =v$ and with $v_{d(v,w)}$ connected to $w$ in the graph obtained by deleting all edges attached to $[v, v_1 ,\ldots , v_{d(v,w)-1}]$.

We define a function $\mathscr{W} (d)$ which measures the maximum number of self-avoiding walks $\mathcal{W}' (v,w)$ with $d$ steps that can happen in $\mathbb{G}$.

$\mathscr{W} (d ) = \max \{  |\mathcal{W}' (v,w)| : d(v,w)=d \}$, thus $|\mathcal{W}' (v,w)| \leq \mathscr{W} (d(v,w))$.\\

We write the set of vertices on a sphere (or shell) with radius $d$ from some origin $v$ as $S(v,d)$. We write the set of vertices in a ball with radius $d$ from some origin $v$ as $B(v,d)$. $|S(v,d)|$ and $|B(v,d)|$ are the number of vertices in $S(v,d)$ and $B(v,d)$ respectively. Then, we define $\mathscr{S}(d)$ to be the largest possible value of $|S(v,d)|$ as $v$ ranges over the graph, $\mathscr{S}(d) = \displaystyle\max_{v} |S(v,d)|$. $\mathscr{S}(d)$ is bounded by the biggest possible value $N(N - 1)^d$.\\

Disordered matter can be described by a \emph{random Schr\"{o}dinger operator} acting on the Hilbert space $l^2 (V)$:
\begin{center}
$l^2 (V) = \{\psi : V \rightarrow \mathbb{C} : \displaystyle\sum_{v \in V} |\psi (v)|^2 < \infty \}$
\end{center}
with inner product $\langle \psi , \phi \rangle = \displaystyle\sum_{v\in V} \bar{\psi} (v) \phi (v)$. 

A random Schr\"{o}dinger operator can be written as
\begin{center}
$H=H_{\omega} = T + \lambda \omega_{v}$
\end{center}
where $T$ is the kinetic energy, the random potential $\omega_{v}$ is a multiplication operator on $l^2 (V)$ with a coupling constant $\lambda >0$. We assume the simplest case where $(\omega_{v})_{v\in V}$ is a set of independent, identically distributed (i.i.d.) real-valued random variables. 

Large coupling constant $\lambda >>1$ indicates large disorder (randomness) and small coupling constant $\lambda <<1$ indicates small disorder. As $\lambda$ increases, the distribution is spread out over larger supports and the random potential can take a wider range of possible random values.\\

We assume that the distribution $\mu$ of $\omega_{v}$ is absolutely continuous with density $\rho$ where $\rho$ is bounded with compact support, i.e.,

\begin{center}
$\mu (B) = \int_{B} \rho  (u) du$ for $B \subset \mathbb{R}$ Borel, \quad $\rho \in L^{\infty}_0 (\mathbb{R})$
\end{center}

Physically, $\omega_{v}$ represents the random electric potential created by nuclei at the sites $v \in V$. $T$ describes the kinetic energy and it is often called next neighbour hopping operator acting on $\psi \in l^2 (V)$. Also, $T$ is the negative adjacency matrix of $\mathbb{G}$.

\begin{center}
$T\psi (v) = - \displaystyle\sum_{w: w\sim v} \psi (w)$
\end{center}
so that

\begin{center}
$(H \psi) (v) = (T \psi) (v) + \lambda \omega_{v} \psi (v) , \quad v \in V$
\end{center}

If we use the Dirac notation, we can write

\begin{center}
$H= -\displaystyle\sum_{\{v,w\}:v\sim w} (|v\rangle\langle w| +|w\rangle\langle v|) + \lambda \displaystyle\sum_{v\in V} \omega_{v} |v\rangle\langle v|$
\end{center}
where $|v\rangle = \delta_{v}$ with $\delta_{v}$ the Kronecker delta function. i.e., $\delta_{v} (v)=1$ and $\delta_{v} (w)=0$ for $v\not= w$. $\{ \delta_{v} \}_{v \in V}$ is the canonical orthonormal basis for $l^2 (V)$. We can write the projection operator as $|v\rangle\langle v| =\langle \delta_{v} , \cdot \rangle \delta_{v}$, where $\langle \cdot , \cdot \rangle$ is the usual scalar product in $l^2 (V)$. For a bounded operator $M$ on $l^2 (V)$, we can write the $(v,w)$-entry of the matrix as $M(v,w) = \langle v|M|w\rangle$.

$T$ is symmetric and bounded since there is an uniform bound $N <\infty$ on the vertex degree. Thus $T$ is self-adjoint. The random potential term $\lambda \omega_{v} : l^2 (V) \rightarrow l^2 (V)$ is symmetric and bounded with the assumption that $\rho$ has compact support. Therefore $H$ is also bounded and self-adjoint.

The quantum mechanical motion of an electron in a disordered system can be described by the above random Schr\"{o}dinger operator $H$ and this model is called Anderson model. 

Anderson localization caused by the absence of electron transport follows from \emph{dynamical localization}.

\begin{definition}
\textbf{\emph{(Dynamical localization)}} $H$ exhibits dynamical localization in $I$ if there exist constants $C <\infty$ and $\mu >0$ such that
\begin{equation}
\displaystyle\sum_{y\in S(x,d)} \mathbb{E} \left( \displaystyle\sup_{t\in\mathbb{R}} |\langle \delta_{y} , e^{-itH} \chi_{I} (H) \delta_{x} \rangle | \right) \leq Ce^{-\mu d}
\end{equation}
 for all $x \in V$.
\end{definition}
$\mathbb{E}$ is the expectation with respect to the probability measure for random variables $\lambda \omega_{v}$. $\chi_{I}$ is the characteristic function of $I$ and so $\chi_{I} (H)$ is the orthogonal projection onto the spectral subspace of $H$ corresponding to energies in $I$. i.e., we only deal with the initial states with energy in $I$.

This is a stronger definition than the standard definition of dynamical localization in the lattice case which requires the expectation for any $x$ and $y$ (with no sum over $y$) to decay exponentially in the distance $d(x,y)$. In this definition, the sum of the expectation over all $y \in S(x,d)$ should decay exponentially with distance $d$. For the lattice, definitions are equivalent since $|S(x,d)|$ grows polynomially.

Dynamical localization gives us physical intuition. It implies that the wavefunctions which are the solutions of the time-dependent Schr\"{o}dinger equation are uniformly localized in space for all times. This leads to the localization of an electron, therefore the absence of electron transport. Furthermore, dynamical localization implies spectral localization \cite{Sto, Su}.

\section{SAW Representation for the Green's Function}
Although the fact that Green's function on a graph can be represented using self-avoiding walks on a graph has been already discussed in \cite{Hun, Tau}, in this section, we derive SAW representation for the Green's function in the way which gives us physical intuition.  

\begin{figure}[htbp]
 \begin{center}
  \includegraphics[width=65mm]{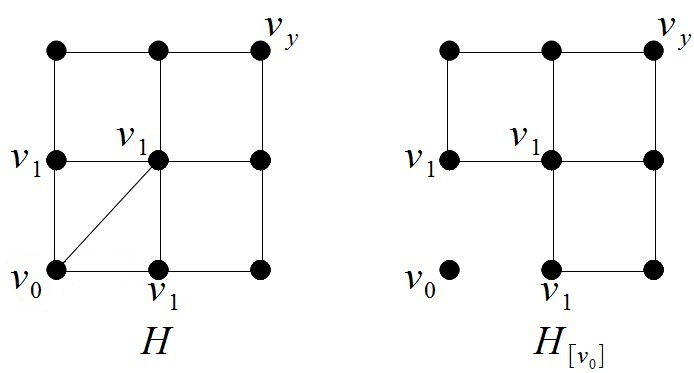}
 \end{center}
 \caption{$H$ and $H_{[v_0]}$}
 \label{fig:one}
\end{figure}

The Green's function $G(x,y;z)$ is the matrix element of the resolvent of $H$, which is written as
\begin{center}
$G(x,y; z) :=\langle x| (H-z)^{-1} |y \rangle$
\end{center}
where $z \in \mathbb{C} \setminus \mathbb{R}$ represents imaginary energy $z= E+i\epsilon$.

We define a depleted random Schr\"{o}dinger operator which can be made from $H$ by self-avoiding walk process as follows (Figure 1).\\

 $H_{[v_0]} = H+ \displaystyle\sum_{v_1 \sim [v_0]} (|v_0\rangle\langle v_1 | + |v_1\rangle\langle v_0|)$ \\
(i.e., $H_{[v_0]}$ is $H$ without edges connected to $v_0$.)\\

$H_{[v_0, \ldots v_{i}]} = H_{[v_0, \ldots v_{i-1}]} + \displaystyle\sum_{v_{i} \sim [v_0, \ldots v_{i-1}]} (|v_{i-1}\rangle\langle v_{i}| + |v_{i}\rangle\langle v_{i-1}|)$.\\
where $\displaystyle\sum_{v_{i} \sim [v_0, \ldots v_{i-1}]}$ is summing over every possible vertex $v_{i}$ which can be reached by taking the next step after the self-avoiding walk $[v_0, \ldots , v_{i-1}]$.\\

Then we obtain the following proposition which is essentially Lemma 4.3 of \cite{Hun}, however we provide a simpler proof here by avoiding convergence issues due to an infinite volume random walk representation used in \cite{Hun}.\\

\begin{proposition}\textbf{\emph{(SAW representation)}} Let $x=v_0,y \in V$. Then the Green's function can be written as\\\\
$G (x, y ;z) = \displaystyle\sum_{[v_0, \ldots , v_{d(x,y)}]} \displaystyle\prod_{i=0}^{d(x,y)-1} \langle v_{i} |(H_{[v_0, \ldots , v_{i-1}]} -z)^{-1} |v_{i} \rangle \langle v_{d(x,y)} | (H_{[v_0, \ldots, v_{d(x,y)-1}]}-z)^{-1} | y\rangle$\\
where $\displaystyle\sum_{[v_0,\ldots,v_{d(x,y)}]}$ is summing over all self-avoiding walks starting at $x$ with length $d(x,y)$. When $i=0$, $H_{[v_0,v_{-1}]} =H$. When $i=1$, $H_{[v_0,v_0]} =H_{[v_0]}$.
\end{proposition}

\begin{proof}

Firstly, $H=H_{[v_0]} -\displaystyle\sum_{v_1 \sim [v_0]} (|v_0\rangle\langle v_1 | + |v_1\rangle\langle v_0 |)$

Using the resolvent formula $A^{-1} - B^{-1} = A^{-1} (B-A) B^{-1}$ with $A= H-z$ and $B= H_{[v_0]} -z$, we have

\begin{center}
$(H-z)^{-1} = (H-z)^{-1} \displaystyle\sum_{v_1 \sim [v_0]} (|v_0\rangle\langle v_1| + |v_1\rangle\langle v_0|)(H_{[v_0]}-z)^{-1}+(H_{[v_0]}-z)^{-1}$
\end{center}

Then, if $d(x,y) \geq 1$ so that $y\not=v_0$,

\begin{center}
$\langle v_0 | (H-z)^{-1} |y\rangle = \langle v_0 | (H-z)^{-1}| v_0\rangle \displaystyle\sum_{v_1 \sim [v_0]} \langle v_1 | (H_{[v_0]} -z)^{-1} |y\rangle$
\end{center}
since $\langle v_0 | (H_{[v_0]} -z)^{-1} |y\rangle =0$ as edges connected to $v_0$ are removed in $H_{[v_0]}$ ($H_{[v_0]}$ is block-diagonal).

Similarly,

\begin{center}
$H_{[v_0]} = H_{[v_0 ,v_1]} -\displaystyle\sum_{v_2 \sim [v_0,v_1]} (|v_1\rangle\langle v_2 | + |v_2\rangle\langle v_1 |)$
\end{center}

Then, if $d(x,y) \geq 2$ so that $y\not= v_1$, we have
\begin{center}
$\langle v_1 | (H_{[v_0]} -z)^{-1} |y\rangle =\langle v_1 | (H_{[v_0]} -z)^{-1} |v_1\rangle \displaystyle\sum_{v_2 \sim [v_0,v_1]} \langle v_2 | (H_{[v_0,v_1]} -z)^{-1} |y\rangle$
\end{center}

Therefore

$\langle v_0 | (H-z)^{-1} |y\rangle = \langle v_0 | (H-z)^{-1}| v_0\rangle \displaystyle\sum_{v_1 \sim [v_0]} \langle v_1 | (H_{[v_0]} -z)^{-1} |v_1\rangle$

\qquad \qquad \qquad \qquad\qquad\qquad\qquad $\times \displaystyle\sum_{v_2 \sim [v_0,v_1]} \langle v_2 | (H_{[v_0,v_1]} -z)^{-1} |y\rangle $

Repeating the above process we have the form:

$G (x, y ;z) = \langle v_0 | (H-z)^{-1} |y\rangle $

$ =  \displaystyle\sum_{v_1 \sim [v_0]}\displaystyle\sum_{v_2 \sim [v_0,v_1]}\cdots  \displaystyle\sum_{v_{d(x,y)} \sim [v_0, \ldots , v_{d(x,y)-1}]} \displaystyle\prod_{i=0}^{d(x,y)-1} \langle v_{i} |(H_{[v_0, \ldots , v_{i-1}]} -z)^{-1} |v_{i} \rangle$

\qquad \qquad \qquad \qquad\quad\qquad \qquad \qquad$\times \langle v_{d(x,y)} | (H_{[v_0, \ldots, v_{d(x,y)-1}]}-z)^{-1} | y\rangle$

$= \displaystyle\sum_{[v_0, \ldots , v_{d(x,y)}]} \displaystyle\prod_{i=0}^{d(x,y)-1} \langle v_{i} |(H_{[v_0, \ldots , v_{i-1}]} -z)^{-1} |v_{i} \rangle \langle v_{d(x,y)} | (H_{[v_0, \ldots, v_{d(x,y)-1}]}-z)^{-1} | y\rangle$\\
since $\displaystyle\sum_{[v_0, \ldots , v_{d(x,y)}]} =  \displaystyle\sum_{v_1 \sim [v_0]}\displaystyle\sum_{v_2 \sim [v_0,v_1]}\cdots  \displaystyle\sum_{v_{d(x,y)} \sim [v_0, \ldots , v_{d(x,y)-1}]}$

\end{proof}
Here, self-avoiding walks are sequences of vertices $[v_0 ,\ldots , v_{d(x,y)}]$. Note that if a walker can not find any edge to walk from $v_{i}$ since he deleted all edges connected to $v_{i}$, then the contribution of that walk to the Green's function is $0$ since  $\langle v_{i} |(H_{[v_0,\ldots , v_{i-1}]} -z)^{-1}|y\rangle =0$ (See $\overline{\mathcal{X}} (x,y)$ in Figure 2).

There may be a small analogy between SAW representation for the Green's function and path integral approach to propagator in quantum mechanics, although the Green's function here is not a propagator.

\section{Fractional Moment Bounds with SAW Representation}

Now, we write the fractional moment bounds of the Green's function in terms of self-avoiding walks, which will be used in the next section. 

Firstly, let $\mathcal{W} (x,y)$ be the set of self-avoiding walks $[x , v_1, \ldots , v_{d(x,y)}]$ with $d(x,y)$ steps starting at $v_0 =x$. Then, we can divide $\mathcal{W} (x,y)$ into three subsets (Figure 2):\\

$\mathcal{Y} (x,y)$ : Self-avoiding walks in $\mathcal{W} (x,y)$ with $v_{d(x,y)} =y$.

$\mathcal{X} (x,y) $ : Self-avoiding walks in $\mathcal{W} (x,y)$ with $v_{d(x,y)} \not=y$ where $v_{d(x,y)}$ is connected to $y$  in the graph obtained by deleting all edges attached to $[x,v_1,\ldots,v_{d(x,y)-1}]$.

$\overline{\mathcal{X}} (x,y) $ : Self-avoiding walks in $\mathcal{W} (x,y)$ with $v_{d(x,y)} \not=y$ where $v_{d(x,y)}$ is not connected to $y$  in the graph obtained by deleting all edges attached to $[x,v_1,\ldots,v_{d(x,y)-1}]$.\\

Only $\mathcal{Y} (x,y)$ and $\mathcal{X} (x,y)$ contribute to the Green's function since
\begin{center}
$\langle v_{d(x,y)} | (H_{[v_0, \ldots, v_{d(x,y)-1}]}-z)^{-1} | y\rangle = 0$\quad if \quad $[x,v_1 , \ldots , v_{d(x,y)}] \in \overline{\mathcal{X}} (x,y).$
\end{center}

We can define $\mathcal{W}' (x,y) = \mathcal{Y} (x,y) \cup \mathcal{X} (x,y)$ be the set of self-avoiding walks $[x,v_1,\ldots,v_{d(x,y)}]$ with $d(x,y)$ steps starting at $v_0 =x$ and $v_{d(x,y)}$ is connected to $y$ in the graph obtained by deleting all edges attached to $[x, v_1 ,\ldots , v_{d(x,y)-1}]$. Note that $y$ is connected to itself in the case $v_{d(x,y)}=y$.

The other fact that we use is an a priori bound:\\

\begin{lemma}\textbf{\emph{(A priori bound)}} Let $0<s<1$. There exist constants $C_1 (s,\rho), C_2 (s, \rho) <\infty$ such that\\

$\mathbb{E}_{x}  (|G (x,x;z)|^{s} )  \leq \|\rho\|^{s}_{\infty} \frac{2^{s} s^{-s}}{1-s} \lambda^{-s} = C_1 (s,\rho) \lambda^{-s}$\\

$\mathbb{E}_{x,y}  (|G(x,y;z)|^{s} )  \leq \|\rho \|^{s}_{\infty} 2^{s+1} \frac{2^{s} s^{-s}}{1-s} \lambda^{-s} = C_2 (s, \rho) \lambda^{-s}$\\ \\
for all $x, y \in V$ where $x \not= y$, $z\in \mathbb{C} \setminus \mathbb{R}$ and $\lambda >0$.\\
Here
\begin{center}
$\mathbb{E}_{x} (\ldots) = \int \ldots \rho (\omega_{x}) d\omega_{x}$
\end{center}
and
\begin{center}
$\mathbb{E}_{x,y} (\ldots) = \int \int \cdots \rho (\omega_{x}) d\omega_{x} \rho (\omega_{y}) d\omega_{y}$
\end{center}
is the conditional expectation with $(\omega_{u})_{u \in V \setminus \{x,y\}}$ fixed. After averaging over $\omega_{x}$ and $\omega_{y}$, the bound does not depend on the remaining random potentials \cite{Sto}.
\end{lemma}
Note $C_1 (s,\rho) < C_2 (s,\rho)$. The proof of Lemma 4.1 is given in \cite{Sto, Tau}.\\

\begin{figure}[htbp]
 \begin{center}
  \includegraphics[width=90mm]{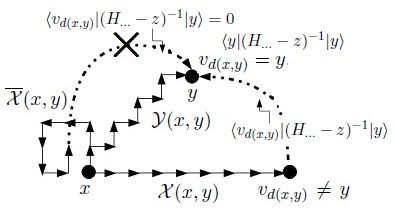}
 \end{center}
 \caption{The difference between three kinds of self-avoiding walks.}
 \label{fig:four}
\end{figure}
\newpage
\begin{theorem}\textbf{\emph{(Fractional moment bounds)}} Let us write the number of walks $\mathcal{Y} (x,y)$, $\mathcal{X} (x,y)$ and $\mathcal{W}' (x,y)$ as $|\mathcal{Y} (x,y)|$, $|\mathcal{X} (x,y)|$ and $|\mathcal{W}' (x,y)|$ respectively. Let $0 < s <1$. Then, the fractional moment bounds of the Green's function can be written as follows.

\begin{center} $\mathbb{E}(|G (x, y ;z)|^{s})\leq |\mathcal{W}' (x,y)| \left( \frac{C_2 (s,\rho)}{\lambda^{s}}\right)^{d(x,y)+1}$\end{center}
where $C_2 (s,\rho)=  \|\rho\|^{s}_{\infty} 2^{s+1}\frac{2^{s}s^{-s}}{1-s}$.\end{theorem}
\begin{proof} 
Let $v_0 =x$.

$\mathbb{E}(|G (x, y ;z)|^{s}) =\mathbb{E} (| \langle v_0 | (H-z)^{-1} |y\rangle |^{s})$

$=\mathbb{E} \left( \left| \displaystyle\sum_{[v_0, \ldots , v_{d(x,y)}]} \displaystyle\prod_{i=0}^{d(x,y)-1} \langle v_{i} |(H_{[v_0, \ldots , v_{i-1}]} -z)^{-1} |v_{i} \rangle \langle v_{d(x,y)} | (H_{[v_0, \ldots, v_{d(x,y)-1}]}-z)^{-1} | y\rangle \right|^{s} \right)$

$\leq \mathbb{E} \left(  \displaystyle\sum_{[v_0, \ldots , v_{d(x,y)}]} \displaystyle\prod_{i=0}^{d(x,y)-1}\left| \langle v_{i} |(H_{[v_0, \ldots , v_{i-1}]} -z)^{-1} |v_{i} \rangle \langle v_{d(x,y)} | (H_{[v_0, \ldots, v_{d(x,y)-1}]}-z)^{-1} | y\rangle \right|^{s} \right)$

$=  \displaystyle\sum_{[v_0, \ldots , v_{d(x,y)}]}\mathbb{E} \left(  \displaystyle\prod_{i=0}^{d(x,y)-1}\left| \langle v_{i} |(H_{[v_0, \ldots , v_{i-1}]} -z)^{-1} |v_{i} \rangle \langle v_{d(x,y)} | (H_{[v_0, \ldots, v_{d(x,y)-1}]}-z)^{-1} | y\rangle \right|^{s} \right)$

$\leq |\mathcal{Y} (x,y) | (C_1 (s,\rho)\lambda^{-s})^{d(x,y)+1} +|\mathcal{X} (x,y) | (C_1 (s,\rho)\lambda^{-s})^{d(x,y)}C_2 (s,\rho)\lambda^{-s}$

$\leq  (|\mathcal{Y}(x,y) |+|\mathcal{X} (x,y)|) ( C_2 (s,\rho)\lambda^{-s})^{d(x,y)+1} =(|\mathcal{W}'(x,y)|) \left( \frac{C_2 (s,\rho)}{\lambda^{s}}\right)^{d(x,y)+1}$\\
where the third step used $|\sum x_{i}|^{s} \leq \sum |x_{i}|^{s}$ for $0<s<1$, the fourth step used the fact that sum of expectations is equal to expectation of sums and the sixth step used $C_1 (s,\rho) < C_2 (s,\rho)$ (they are different just by a factor $2^{s+1}$).  In the fifth step, we have\\

$\mathbb{E} \left(  \displaystyle\prod_{i=0}^{d(x,y)-1}\left| \langle v_{i} |(H_{[v_0, \ldots , v_{i-1}]} -z)^{-1} |v_{i} \rangle \langle v_{d(x,y)} | (H_{[v_0, \ldots, v_{d(x,y)-1}]}-z)^{-1} | y\rangle \right|^{s} \right)$\\

$=\mathbb{E}  (| \langle v_{0} |(H -z)^{-1} |v_{0} \rangle |^{s} |\langle v_{1} |(H_{[v_0]} -z)^{-1} |v_{1} \rangle |^{s} \times \cdots\times \\ \qquad \qquad | \langle v_{d(x,y)-1} |(H_{[v_0, \ldots , v_{d(x,y)-2}]} -z)^{-1} |v_{d(x,y)-1} \rangle |^{s} | \langle v_{d(x,y)} | (H_{[v_0, \ldots, v_{d(x,y)-1}]}-z)^{-1} | y\rangle |^{s} )$\\

$=\mathbb{E}_{v\setminus v_0} [ \mathbb{E}_{v_0} (| \langle v_{0} |(H -z)^{-1} |v_{0} \rangle |^{s}) \times |\langle v_{1} |(H_{[v_0]} -z)^{-1} |v_{1} \rangle |^{s} \times \cdots\times $

\qquad\quad$| \langle v_{d(x,y)-1} |(H_{[v_0, \ldots , v_{d(x,y)-2}]} -z)^{-1} |v_{d(x,y)-1} \rangle |^{s} | \langle v_{d(x,y)} | (H_{[v_0, \ldots, v_{d(x,y)-1}]}-z)^{-1} | y\rangle |^{s} ]$
\begin{flushright}(4.1)\end{flushright}
where $\mathbb{E}_{v\setminus v_0}$ is the expectation with respect to every random potential except the one at $v_0$ which is $\omega_{v_0}$ and $\mathbb{E}_{v_0}$ is the expectation with respect to $\omega_{v_0}$.

Since only $| \langle v_{0} |(H -z)^{-1} |v_{0} \rangle |^{s}$ depends on $\omega_{v_0}$, by Lemma 4.1, \\

(4.1) $\leq C_1 (s,\rho) \lambda^{-s} \mathbb{E}_{v\setminus v_0} (|\langle v_{1} |(H_{[v_0]} -z)^{-1} |v_{1} \rangle |^{s} \times \cdots\times$

\qquad\quad$| \langle v_{d(x,y)-1} |(H_{[v_0, \ldots , v_{d(x,y)-2}]} -z)^{-1} |v_{d(x,y)-1} \rangle |^{s} | \langle v_{d(x,y)} | (H_{[v_0, \ldots, v_{d(x,y)-1}]}-z)^{-1} | y\rangle |^{s})$\\

$=C_1 (s,\rho) \lambda^{-s} \mathbb{E}_{v\setminus \{v_0, v_1\} } [\mathbb{E}_{v_1}(|\langle v_{1} |(H_{[v_0, v_1]} -z)^{-1} |v_{1} \rangle |^{s}) \times \cdots\times$

\qquad\quad$| \langle v_{d(x,y)-1} |(H_{[v_0, \ldots , v_{d(x,y)-2}]} -z)^{-1} |v_{d(x,y)-1} \rangle |^{s} | \langle v_{d(x,y)} | (H_{[v_0, \ldots, v_{d(x,y)-1}]}-z)^{-1} | y\rangle |^{s}]$\\

$\leq  C_1^2 (s,\rho) \lambda^{-2s} \mathbb{E}_{v\setminus \{v_0, v_1, v_2\} } [\mathbb{E}_{v_2}(|\langle v_{2} |(H_{[v_0, v_1, v_2]} -z)^{-1} |v_{2} \rangle |^{s}) \times \cdots\times$

\qquad\quad$| \langle v_{d(x,y)-1} |(H_{[v_0, \ldots , v_{d(x,y)-2}]} -z)^{-1} |v_{d(x,y)-1} \rangle |^{s} | \langle v_{d(x,y)} | (H_{[v_0, \ldots, v_{d(x,y)-1}]}-z)^{-1} | y\rangle |^{s}]$

Here, we used the fact that Lemma 4.1 can also be applied to $H_{[v_0 , v_1 ,\ldots]}$.

Continuing the process, we obtain the above expression. Note that, in the last step, we have\\

$\mathbb{E}_{y}  (|\langle y|(H_{[v_0,\ldots,v_{d(x,y)-1}]} |y \rangle |^{s} ) \leq  C_1 (s,\rho) \lambda^{-s}$\quad if \quad$d(x,y) = y$.\\

$\mathbb{E}_{v_{d(x,y)},y}  (|\langle v_{d(x,y)}|(H_{[v_0,\ldots,v_{d(x,y)-1}]} |y \rangle |^{s} ) \leq C_2 (s, \rho) \lambda^{-s}$\quad if \quad$d(x,y) \not= y$.

\end{proof}

 In the lattice case, FMM states that if the fractional moment bounds decay sufficiently rapidly, we obtain dynamical localization. In the general graph case, the situation is more complicated as we will discuss in the next section. However, it is still true in the general graph case that as the number of self-avoiding walks increases, the system needs larger $\lambda$ to obtain dynamical localization. Even the type of graph changes, if $|\mathcal{W}' (x,y) |$ stays the same, the fractional moment bounds of the Green's function does not change. One-dimensional lattice and tree graphs such as Bethe lattice have $|\mathcal{Y} (x,y) |=1$ and $|\mathcal{X} (x,y)| =0$ since only one self-avoiding walk can arrive $y$ with $d(x,y)$ steps and the other self-avoiding walks with $d(x,y)$ steps do not have edges connected to $y$. This implies that one-dimensional Euclidean lattice and tree graphs have the same fractional moment bounds. However, as we will see in the next section, they still behave differently on dynamical localization.

\section{Dynamical Localization and Graph Properties}
In this section, we introduce FMM which connects the fractional moment bounds of the Green's function and dynamical localization in the large disorder regime. 

Recall that we defined two functions in the section 2.\\

$\mathscr{W} (d ) = \max \{  |\mathcal{W}' (x,y)| : d(x,y)=d \}$, thus $|\mathcal{W}' (x,y)| \leq \mathscr{W} (d(x,y))$.

$\mathscr{S} (d) = \displaystyle\max_{x} | S(x,d)|$.

\begin{theorem} Let $I \subset \mathbb{R}$, $0<s<1$ and $\epsilon \in (0,1/2)$. Then dynamical localization (1) holds for disorder $\lambda$ which satisfies the following condition. 
\begin{center}
 $ s\epsilon \ln \lambda > \epsilon\ln C_2 + \displaystyle\sup_{x,d} \left( \frac{\ln |S(x,d)|}{d}\right) $\quad and \quad $\displaystyle\sum_{d'=0}^{\infty} \mathscr{S} (d') \mathscr{W} (d')  \left(\frac{C_2}{\lambda^{s}}\right)^{(1-2\epsilon)d'}<\infty$
\end{center}

\medskip

\end{theorem}

The proof uses the argument by Graf \cite{Gr} that the fractional moment of Green's function $\mathbb{E}_{x} (|G (x,y;z)|^{s})$ can bound the second moment of Green's function $\mathbb{E}_{x} (|G (x,y;z)|^2)$:\\

\begin{proposition}For every $0< s <1$, there exists a constant $C'<\infty$ only depending on $s$ and $\rho$ such that
\begin{center}
$|\emph{\mbox{Im}} z| \mathbb{E}_{x} (|G(x,y;z)|^2) \leq C' \mathbb{E}_{x} (|G(x,y;z)|^{s})$
\end{center}
for all $z\in \mathbb{C} \setminus \mathbb{R}$ and $x,y \in \mathbb{G}$. $\mathbb{E}_{x}$ denotes averaging over $\omega_{x}$. \end{proposition}

$C'$ is a constant which appear in Proposition 5.1 of \cite{Sto}. The proof of Proposition also can be found in \cite{Gr, Sto}. \\

Now we prove Theorem 5.1 using Proposition 5.2. The proof of Theorem 5.1 uses a method which is well known
for the Anderson model on $\mathbb{Z}^{n}$ and can be applied directly for the Anderson model on graphs.

\begin{proof}

Firstly, we follow the proof in \cite{Gr, Sto} which introduces the complex Borel spectral measures $\mu_{y,x}$ of $H$ written as

\begin{center}
$\mu_{y,x} (B) = \langle \delta_{y}, \chi_{B} (H) \delta_{x}\rangle $
\end{center}
for Borel sets $B\subset\mathbb{R}$. Then, the total variation $|\mu_{y,x}|$ of $\mu_{y,x}$ is given by

\begin{center}
$|\mu_{y,x}| (B) = \displaystyle\sup_{g:\mathbb{R}\rightarrow\mathbb{C}, Borel, |g| \leq 1} \left| \int_{B} g (\lambda) d\mu_{y,x} (\lambda) \right| =\displaystyle\sup_{|g|\leq 1} |\langle \delta_{y} , g(H) \chi_{B} (H) \delta_{x}\rangle |$

\end{center}

This is a regular bounded Borel measure.

If we choose $g(H)=e^{-itH}$, we can bound the expectation in (1).

\begin{center}
$\displaystyle\sum_{y\in S(x,d)}\mathbb{E}\left(\displaystyle\sup_{t\in\mathbb{R}} |\langle \delta_{y}, e^{-itH} \chi_{I} (H) \delta_{x}\rangle| \right) \leq \displaystyle\sum_{y\in S(x,d)} \mathbb{E}(|\mu_{y,x}| (I)) $
\end{center}

Therefore, the exponential decay of $\displaystyle\sum_{y\in S(x,d)} \mathbb{E}(|\mu_{y,x}| (I))$ implies dynamical localization. In the same way as \cite{Sto}, we have\\

$\displaystyle\sum_{y\in S(x,d)} \mathbb{E} (|\mu_{y,x} |(I))$

$\leq \displaystyle\sum_{y\in S(x,d)}\mathbb{E} \left( \displaystyle\liminf_{\epsilon\rightarrow 0+} \frac{\epsilon}{\pi} \int_{I} \displaystyle\sum_{z\in\mathbb{G}} |\langle \delta_{y} , (H-E-i\epsilon)^{-1} \delta_{z}\rangle || \langle \delta_{z}, (H-E+i\epsilon)^{-1} \delta_{x}\rangle | dE \right)$

$\leq \frac{1}{\pi} \displaystyle\liminf_{\epsilon\rightarrow 0+}\displaystyle\sum_{y\in S(x,d)} \mathbb{E} \left(  \epsilon \int_{I} \displaystyle\sum_{z\in\mathbb{G}} |\langle \delta_{y} , (H-E-i\epsilon)^{-1} \delta_{z}\rangle || \langle \delta_{z}, (H-E+i\epsilon)^{-1} \delta_{x}\rangle | dE \right)$

$= \frac{1}{\pi} \displaystyle\liminf_{\epsilon\rightarrow 0+}\int_{I} \displaystyle\sum_{z\in\mathbb{G}}\displaystyle\sum_{y\in S(x,d)} \mathbb{E} \left(  \epsilon  |G(y,z;E+i\epsilon) || G(z,x;E+i\epsilon) |  \right)dE$

$\leq \displaystyle\liminf_{\epsilon\rightarrow 0+} \frac{1}{\pi} \int_{I} \displaystyle\sum_{z\in\mathbb{G}} \displaystyle\sum_{y\in S(x,d)}(\mathbb{E} (\epsilon |G(y,z ; E+i\epsilon) |^2))^{1/2} \cdot (\mathbb{E} (\epsilon |G(z,x ; E+i\epsilon)  |^2 ))^{1/2} dE$

The second step used Fatou's lemma, the third step used Fubini's theorem and the fourth step used Cauchy-Schwarz inequality. Now, we introduce Proposition 5.2 and Theorem 4.2.\\

$\displaystyle\sum_{y\in S(x,d)} \mathbb{E} (|\mu_{y,x} |(I))$

$\leq \displaystyle\lim_{\epsilon\rightarrow 0+} \frac{C'}{\pi} \int_{I} \displaystyle\sum_{z\in\mathbb{G}} \displaystyle\sum_{y\in S(x,d)}(\mathbb{E} (|G(y,z; E+i\epsilon )|^{s}))^{1/2} (\mathbb{E} (|G(z,x;E-i\epsilon)|^{s}))^{1/2} dE$

$\leq \frac{C'  |I|}{\pi} \displaystyle\sum_{z\in\mathbb{G}}\displaystyle\sum_{y\in S(x,d)} |\mathcal{W}' (y,z)|^{1/2} \left(\frac{C_2}{\lambda^{s}}\right)^{\frac{d(y,z)+1}{2}} |\mathcal{W}' (z,x) |^{1/2}  \left(\frac{C_2}{\lambda^{s}}\right)^{\frac{d(z,x)+1}{2}}$\\

Using triangle inequality,\\

$\frac{d(y,z) + d(z,x)}{2} = \epsilon (d(y,z) + d(z,x)) + \left(\frac{1}{2} -\epsilon \right)(d(y,z) + d(z,x))$

\qquad\qquad\quad $\geq \epsilon d(y,x)+ \left(\frac{1}{2} -\epsilon \right)(d(y,z) + d(z,x))$\\
where $0 <\epsilon < 1/2$. Then, by assuming $\frac{C_2}{\lambda^{s}} <1$ ($\lambda^{s} > C_2$), we have\\

$\left(\frac{C_2}{\lambda^{s}}\right)^{ (d(y,z)+d(z,x))/2} \leq \left(\frac{C_2}{\lambda^{s}}\right)^{\epsilon d(y,x) + (1/2 -\epsilon) (d(y,z) + d(z,x))}$ \\

Therefore,\\

 $\displaystyle\sum_{y\in S(x,d)} \mathbb{E} (|\mu_{y,x} |(I)) $

$\leq \frac{C'  |I|}{\pi}\displaystyle\sum_{z\in\mathbb{G}} \displaystyle\sum_{y\in S(x,d)} \left(\frac{C_2}{\lambda^{s}}\right)^{\epsilon d(y,x)}|\mathcal{W}' (y,z)|^{1/2}  \left(\frac{C_2}{\lambda^{s}}\right)^{(1/2 -\epsilon) d(y,z)+1/2}  |\mathcal{W}' (z,x) |^{1/2}  \left(\frac{C_2}{\lambda^{s}}\right)^{(1/2 -\epsilon) d(z,x)+1/2}$ 

$= \frac{C'  |I|}{\pi}\left(\frac{C_2}{\lambda^{s}}\right)^{\epsilon d}\displaystyle\sum_{z\in\mathbb{G}} \displaystyle\sum_{y\in S(x,d)} |\mathcal{W}' (y,z)|^{1/2}  \left(\frac{C_2}{\lambda^{s}}\right)^{(1/2 -\epsilon) d(y,z)+1/2} |\mathcal{W}' (z,x) |^{1/2}  \left(\frac{C_2}{\lambda^{s}}\right)^{(1/2 -\epsilon) d(z,x)+1/2}$ 

$\leq \frac{C'  |I|}{\pi}\left(\frac{C_2}{\lambda^{s}}\right)^{\epsilon d}\displaystyle\sum_{y\in S(x,d)} \left( \displaystyle\sum_{z\in\mathbb{G}} |\mathcal{W}' (y,z)|  \left(\frac{C_2}{\lambda^{s}}\right)^{(1-2\epsilon)d(y,z)+1} \right)^{1/2}  \left( \displaystyle\sum_{z\in\mathbb{G}}|\mathcal{W}' (z,x)|  \left(\frac{C_2}{\lambda^{s}}\right)^{(1-2\epsilon)d(z,x)+1}\right)^{1/2}$\\

by Cauchy-Schwarz inequality.\\

$\displaystyle\sum_{y\in S(x,d)} \mathbb{E} (|\mu_{y,x} |(I))$

$\leq \frac{C'  |I|}{\pi}\left(\frac{C_2}{\lambda^{s}}\right)^{\epsilon d}\displaystyle\sum_{y\in S(x,d)} \left( \displaystyle\sum_{z\in\mathbb{G}}\mathscr{W} (d(y,z))  \left(\frac{C_2}{\lambda^{s}}\right)^{(1-2\epsilon)d(y,z)+1} \right)^{1/2} \left( \displaystyle\sum_{z\in\mathbb{G}}\mathscr{W} (d(z,x))  \left(\frac{C_2}{\lambda^{s}}\right)^{(1-2\epsilon)d(z,x)+1}\right)^{1/2}$

 $= \frac{C'  |I|}{\pi}\left(\frac{C_2}{\lambda^{s}}\right)^{\epsilon d}\displaystyle\sum_{y\in S(x,d)}\left(\displaystyle\sum_{d'=0}^{\infty}\displaystyle\sum_{z\in S (y,d')}\mathscr{W} (d')  \left(\frac{C_2}{\lambda^{s}}\right)^{(1-2\epsilon)d'+1} \right)^{1/2}  \left( \displaystyle\sum_{d'=0}^{\infty}\displaystyle\sum_{z\in S (x,d')}\mathscr{W} (d')  \left(\frac{C_2}{\lambda^{s}}\right)^{(1-2\epsilon)d'+1}\right)^{1/2}$

 $= \frac{C'  |I|}{\pi}\left(\frac{C_2}{\lambda^{s}}\right)^{\epsilon d}|S(x,d)| \left(\displaystyle\sum_{d'=0}^{\infty} |S (y,d')| \mathscr{W} (d')  \left(\frac{C_2}{\lambda^{s}}\right)^{(1-2\epsilon)d'+1} \right)^{\frac{1}{2}}  \left( \displaystyle\sum_{d'=0}^{\infty} |S (x,d')| \mathscr{W} (d')  \left(\frac{C_2}{\lambda^{s}}\right)^{(1-2\epsilon)d'+1}\right)^{\frac{1}{2}} $

$ \leq  \frac{C'  |I|}{\pi}\left(\frac{C_2}{\lambda^{s}}\right)^{\epsilon d}|S(x,d)| \displaystyle\sum_{d'=0}^{\infty} \mathscr{S} (d') \mathscr{W} (d')  \left(\frac{C_2}{\lambda^{s}}\right)^{(1-2\epsilon)d'+1}  \leq Ce^{-\mu d}$\\
where $\mu= s\epsilon \ln \lambda-\epsilon\ln C_2 -\displaystyle\sup_{x,d} \left( \frac{\ln |S(x,d)|}{d}\right) $ and $C=\frac{C'C_2 |I|}{\pi\lambda^{s}}\displaystyle\sum_{d'=0}^{\infty} \mathscr{S} (d') \mathscr{W} (d')  \left(\frac{C_2}{\lambda^{s}}\right)^{(1-2\epsilon)d'}$.

Therefore we have dynamical localization (1) if \\

$\mu >0$, thus, $ s\epsilon \ln \lambda > \epsilon\ln C_2 +\displaystyle\sup_{x,d} \left( \frac{\ln |S(x,d)|}{d}\right) $ and

$\displaystyle\sum_{d'=0}^{\infty} \mathscr{S} (d') \mathscr{W} (d')  \left(\frac{C_2}{\lambda^{s}}\right)^{(1-2\epsilon)d'} <\infty$. 
\end{proof}

This indicates that although the trees and one-dimensional lattice have the
same fractional moment bounds, trees need larger $\lambda$ to obtain dynamical
localization because of the factors $|S(x,d)|$ and $\mathscr{S}(d')$.
\section{Conclusions \& Discussion}

One of the most important open problems in random operator theory is to understand the transition between the  localized regime and the extended states regime. There is an attempt to understand the transition using the \emph{level statistics conjecture} and \emph{random matrix theory} (RMT). This method allows physicists to distinguish the two regimes numerically.

We use the statistical distribution of the eigenvalues of finite volume restrictions of the Anderson model to distinguish two regimes. It is expected that the localized regime and the extended states regime are corresponding to Poisson statistics and Gaussian orthogonal ensemble (GOE) statistics of the eigenvalues respectively.

Some studies proved mathematically that the finite volume eigenvalues show Poisson distribution in the localized regime. Molchanov first proved the Poisson statistics for eigenvalues for one-dimensional continuum random Schr\"{o}dinger operator \cite{Mo}. Subsequently, Minami \cite{Mi} proved Poisson statistics for eigenvalues of the Anderson model. He assumed the exponential decay of the fractional moment of the Green's function holds for complex energies near $E$. Then, he proved the random sequence of rescaled eigenvalues of finite volume converges weakly to the stationary Poisson point process as the finite volume gets large and there is no correlation between eigenvalues near the energy $E$ where Anderson localization is expected.

However, it is still an open problem whether the extended states regime can be characterized by GOE statistics. In RMT, GOE statistics can be obtained for \emph{Wigner random matrices}. All elements in Wigner matrices are random, while only the diagonal matrix elements are random in the Anderson model. Therefore, it is suggested that \emph{random band matrices} which increase amount of off-diagonal random entries can be an useful tool to understand the transition between two regimes \cite{Sto}.

Some studies in the physics community make use of such method to test the transition. Although the study of the Anderson model started in the condensed matter physics, recent studies show this model has an application in many areas such as quantum computing and quantum biology because of the possibility of building the systems using condensed matter physics. 

In quantum computing, Giraud et al. \cite{Gi} studied the model of a circular graph with on-site disorder where each vertex is linked with its two nearest-neighbours and also they added shortcut edges between random pairs of vertices (Figure 3). Therefore, this is the one-dimensional Anderson model with extra off-diagonal random elements.

\begin{figure}[htbp]
 \begin{center}
  \includegraphics[width=35mm]{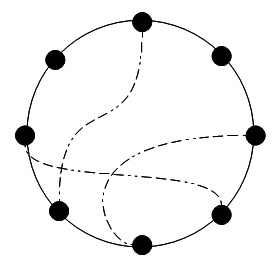}
 \end{center}
 \caption{The dashed lines describe random shortcut edges which represent the off-diagonal disorder \cite{Gi}.}
 \label{fig:two}
\end{figure}

They studied level spacing statistics for Hamiltonian and obtained GOE distribution for small on-site disorder $\lambda$ and Poisson distribution as they made on-site disorder $\lambda$ larger. According to level statistics conjecture, it is expected that GOE distribution represents the extended states, while Poisson distribution represents the localized states.

It might be possible to make a relation between this transition and Theorem 5.1. Adding
off-diagonal random entries (shortcut edges between random pairs of vertices) increases
the number of self-avoiding walks $\mathscr{W}(d)$ and firstly we have the extended states with small on-site disorder. As we increase on-site disorder $\lambda$, it overcomes the number of self-avoiding walks and we obtain the localized states.

A small number of self-avoiding walks may correspond to the localized states by
Theorem 5.1 and also Poisson statistics since distant regions are uncorrelated and the
system creates almost independent eigenvalues which do not have energy repulsion. On
the other hand, a larger number of self-avoiding walks may correspond to the delocalized
states by Theorem 5.1 if $\lambda$ is not large enough to overcome $\mathscr{W} (d)$ and also GOE statistics since distant
regions are correlated, which creates energy level repulsion \cite{Com}. When $\lambda$ overcomes $\mathscr{W}(d)$,
we may have the transition from the extended states to the localized states.

In our work, the connection between distant regions is reflected in the size of $\mathscr{W} (d)$. When $\mathscr{W} (d)$ is large, it is harder to obtain dynamical localization. Also, different from diagonal disorder $\lambda$, off-diagonal disorder
does not always work for localization, but it works against localization when it increases
the number of self-avoiding walks $\mathscr{W}(d)$.

\ack{}
This paper is based on part of the author's master thesis. I am deeply grateful to my advisors Richard Froese and P. C. E. Stamp for their guidance and support. It is a pleasure to thank Joel Feldman for giving me detailed comments and suggestions.


\section*{References}
\begin{thebibliography}{10}

\bibitem{Aizen1} Aizenman M 1994 \emph{ Localization at weak disorder: Some elementary bounds.} \emph{Rev. Math. Phys.} \textbf{6} 1163-1182.
\bibitem{Aizen} Aizenman M \& Molchanov S 1993 \emph{ Localization at large disorder and at extreme energies: an
elementary derivation.} \emph{Comm. Math. Phys.} \textbf{157} 245-278.

\bibitem{Com} Combes G, Germinet F \& Klein A 2009 \emph{Poisson Statistics for Eigenvalues of Continuum Random Schr\"{o}dinger Operators.} [arXiv:0807.0455]

\bibitem{Cy} Cycon H. L, Froese R. G, Kirsch W \& Simon B 1987 \emph{Schr\"{o}dinger Operators with Application to Quantum Mechanics and Global Geometry.} Texts and Monographs in Physics Springer.

\bibitem{Fr} Fr\"{o}hlich J \& Spencer T 1983 \emph{Absence of diffusion in the Anderson tight binding model for large disorder or low energy.} \emph{Comm. Math. Phys.} \textbf{151} 184.

\bibitem{Gi} Giraud O, Georgeot B \& Shepelyansky D.L 2005 \emph{Quantum computing of delocalization in small-world networks.} \emph{Phys. Rev. E.} \textbf{72} 036203.

\bibitem{Gr} Graf G. M 1994 \emph{Anderson localization and the space-time characteristic of continuum states.} \emph{J.
Stat. Phys.} \textbf{75} 337-346.


\bibitem{Hun} Hundertmark D 2008 \emph{A short introduction to Anderson localization.} In Analysis and
Stochastics of Growth Processes and Interface Models, Oxford Scholarship Online Monographs, 194-219.

\bibitem{Kl} Klein A 2008 \emph{Multiscale Analysis and Localization of random operators.} \emph{Panoramas et synth\`{e}ses.} \textbf{25} 121-159.

\bibitem{Mi} Minami N 1996 \emph{Local fluctuation of the spectrum of a multidimensional Anderson tight binding
model.} \emph{Comm. Math. Phys.} \textbf{177} 709-725.

\bibitem{Mo}Molchanov S. A 1981 \emph{The local structure of the spectrum of the one-dimensional Schr\"{o}dinger
operator.} \emph{Comm. Math. Phys.} \textbf{78} 429-446.



\bibitem{Sto} Stolz G 2010 \emph{An introduction to the mathematics of Anderson localization.} Lecture notes of the Arizona School of Analysis with Applications.

\bibitem{Su} Suzuki F 2012 \emph{Anderson Localization with Self-Avoiding Walk Representation} MSc thesis, UBC.


\bibitem{Tau} Tautenhahn M 2011 \emph{Localization criteria for Anderson models on locally finite graphs.} \emph{J.
Stat. Phys.} \textbf{144} 60-75.
\end{thebibliography}
\end{document}